\documentclass{elsarticle}
\usepackage[english]{babel}
\usepackage[utf8]{inputenc}
\usepackage[T1]{fontenc}
\usepackage[fleqn]{amsmath}
\usepackage{amssymb}
\usepackage{amsthm}
\usepackage{amsfonts}
\usepackage{ae}
\usepackage{tikz}
\usepackage{enumerate}
\usepackage{stmaryrd}
\usepackage{fullpage}
\usetikzlibrary{decorations.pathreplacing}

\everymath{\displaystyle} 
\newtheorem{theorem}{Theorem}
\newtheorem{lemma}[theorem]{Lemma}
\newtheorem{prop}[theorem]{Proposition}
\newtheorem{claim}[theorem]{Claim}

\newtheorem{question}[theorem]{Question}

\DeclareRobustCommand{\rchi}{{\mathpalette\irchi\relax}} 
\newcommand{\irchi}[2]{\raisebox{\depth}{$#1\chi$}} 

\newcommand{\kiu}{\rchi_{\cup}} 
\newcommand\id[1]{id(#1)} 
\newcommand\idC[2]{id_{#1}(#2)} 
\newcommand\logV[1]{\lceil \log_2(|#1|+1) \rceil} 
\newcommand\logN{\lceil \log_2(n+1) \rceil}

\tikzstyle{noeud}=[circle, fill=black, inner sep= 0, minimum size = 4]

\title{A Vizing-like theorem for union vertex-distinguishing edge coloring}
\author[gscop]{Nicolas Bousquet}
\author[liris]{Antoine Dailly\corref{cor}}
\author[liris]{Éric Duchêne}
\author[liris]{Hamamache Kheddouci}
\author[liris]{Aline Parreau}
\date{}
\cortext[cor]{Corresponding author}
\address[gscop]{G-SCOP (CNRS, Univ. Grenoble-Alpes), Grenoble, France.}
\address[liris]{Univ Lyon, Université Lyon 1, LIRIS UMR CNRS 5205, F-69621, Lyon, France.}

\begin{document}


\begin{frontmatter}
\begin{abstract}
We introduce a variant of the vertex-distinguishing edge coloring problem, where each edge is assigned a subset of colors. The label of a vertex is the union of the sets of colors on edges incident to it. In this paper we investigate the problem of finding a coloring with the minimum number of colors where every vertex receives a distinct label. Finding such a coloring generalizes several other well-known problems of vertex-distinguishing colorings in graphs.

We show that for any graph (without connected component reduced to an edge or a single vertex), the minimum number of colors for which such a coloring exists can only take $3$ possible values depending on the order of the graph. Moreover, we provide the exact value for paths, cycles and complete binary trees.
\end{abstract}

\end{frontmatter}

\section{Introduction}
Vertex-distinguishing edge colorings of graphs is a wide studied field in chromatic graph theory. Generally speaking, it consists in an edge coloring of a graph (not necessarily proper) that leads to a vertex labeling where every pair of vertices of the graph are distinguished by their labels (also called {\em codes}). For instance, a {\em set irregular edge coloring}~\cite{Harary} is an edge coloring of a graph where each vertex $v$ is assigned the set of colors of the edges incident to $v$. In the literature, many other variants were considered where the codes are defined by the multisets of the colors incident to $v$~\cite{aigner}, the sums\footnote{In that case, colors are integers.}~\cite{chartrand}, the products, or differences~\cite{tahraouiDK12}. The case where the edge coloring is proper was also considered by Burris and Schelp~\cite{Burris}.

More recently, several other variants of this problem were defined, in which the codes are produced from a vertex coloring. In the literature, they are refereed to as {\em identifying colorings}. More precisely, from a vertex coloring of a graph, the code of a vertex is defined as the set of colors of its extended neighborhood. Instances of the identifying coloring problem where any pair of distinct (resp. adjacent) vertices must have different codes were introduced in~\cite{these} (resp.~\cite{Louis}). If few results are known about this problem, its interest is growing as it is a natural generalization of the famous identifying code problem~\cite{codes} where the code of a vertex is given by its presence in the neighbourhood of a set of a vertices.

In the current paper, our objective is to generalize both, the set irregular and identifying coloring problems. For that purpose, we define a new vertex-distinguishing edge coloring problem where every edges is assigned a subset of colors.
Given a simple graph $G$, a $k$-\emph{coloring} of $G$ is a function $f:E(G) \rightarrow 2^{\{1\ldots,k\}}$ where every edge is labeled using a non-empty subset of $\{1,\ldots,k\}$. For any $k$-coloring $f$ of $G$, we define, for every vertex $u$, the set $\idC{f}{u}$ as follows:
$$
 \idC{f}{u} = \bigcup_{v \ s.t. \ uv \in E} f(uv).
$$
If the context is clear, we will simply write $\id{u}$ for $\idC{f}{u}$.
A $k$-coloring $f$ is \emph{union vertex-distinguishing} if, for all distinct $u,v$ in $V(G)$, $\idC{f}{u} \neq \idC{f}{v}$. For a given graph $G$, we denote by $\kiu(G)$ the smallest integer $k$ such that there exists a union vertex-distinguishing coloring of $G$. Figure~\ref{fig:intro} illustrates such a coloring, where small labels correspond to the $f$ function, and the large labels denote the codes.

 \begin{figure}[!h]
 \begin{center}
 \includegraphics[scale=.8]{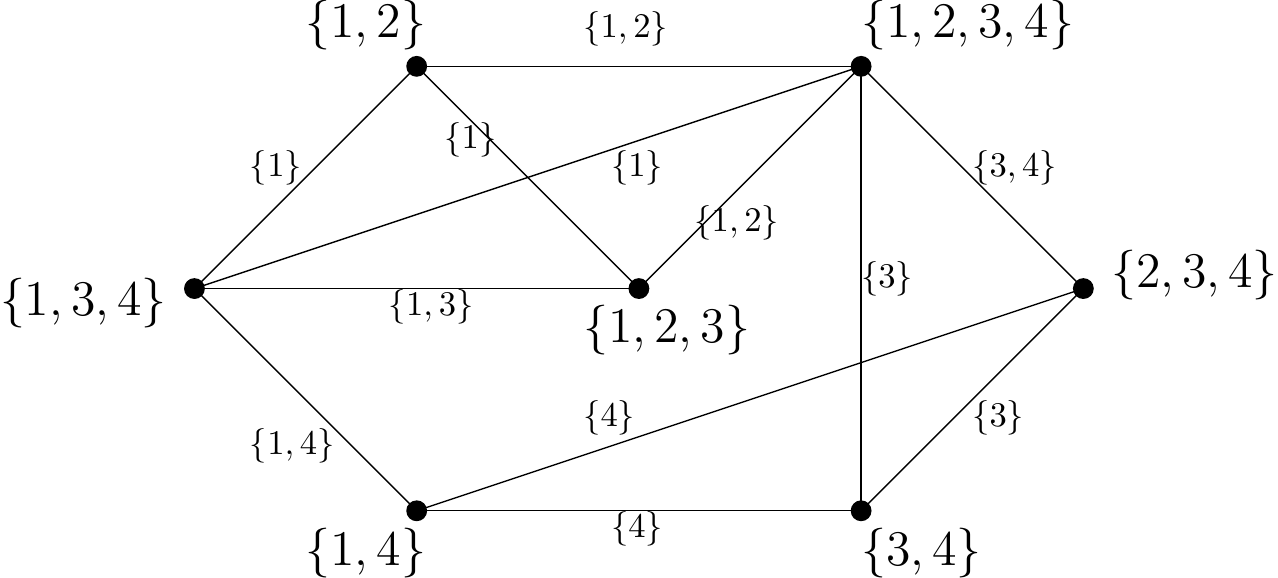}
 \caption{\label{fig:intro} Example of a union vertex-distinguishing coloring }
 \end{center}
 \end{figure}

In most of the existing variants, only one color is allocated to each edge or vertex. Note however that this idea of a coloring function that maps to subsets of integers was also considered by Hegde in 2009. In~\cite{Hegde}, Hedge defined an edge-distinguishing vertex coloring where vertices are colored with sets of positive integers. The codes are defined on edges and equal the symmetric difference of the sets of the vertices incident to them.

As already mentioned, the union vertex-distinguishing problem generalizes existing problems related to vertex-distinguishing colorings. For instance, the set irregular edge coloring problem can be seen as an instance of the union vertex-distinguishing problem where only singletons are allowed on the edges. Moreover, any identifying coloring $f$ of a graph $G$ induces a valid union vertex-distinguishing coloring. Indeed, it suffices to color each edge of $G$ with the set of its incident colors in $f$. Figure~\ref{fig:intro2} illustrates this transformation. As a consequence, each edge of $G$ is colored with a set of size at most $2$.

 \begin{figure}[!h]
 \begin{center}
 \includegraphics[scale=.6]{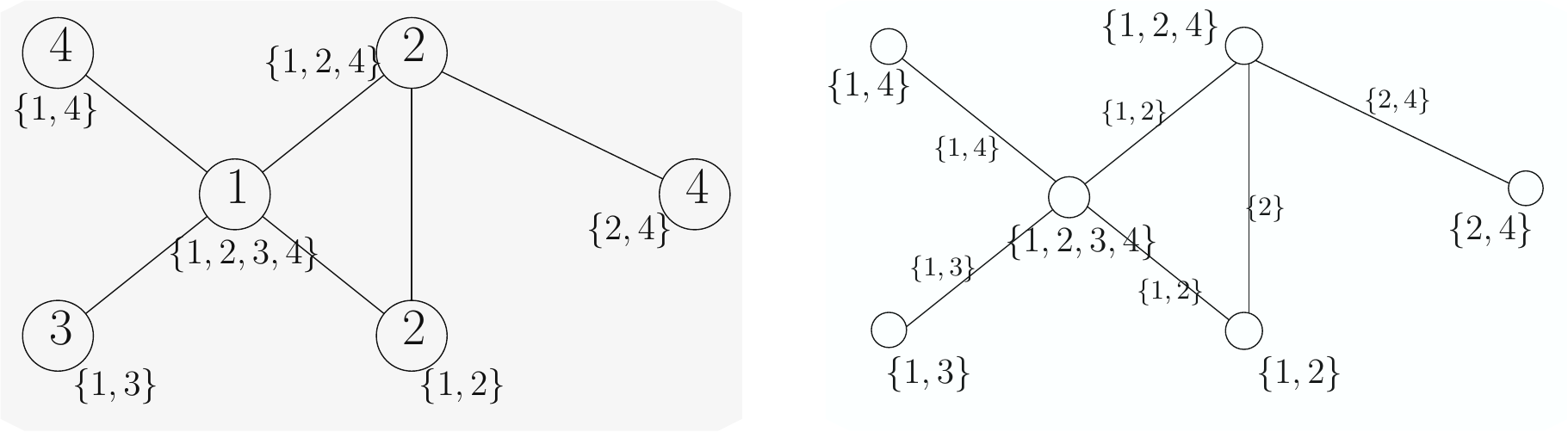}
 \caption{Union vertex-distinguishing coloring (right side) produced from an identifying coloring (left side)\label{fig:intro2}}
 \end{center}
 \end{figure}

Thus if $\rchi_S(G)$ (resp. $\rchi_{id}(G)$) denotes the minimum number of colors of a set irregular edge (resp. identifying) coloring of a graph $G$, the following inequality holds:
\begin{eqnarray}
 \kiu(G)&\leq& \min(\rchi_S(G),\rchi_{id}(G)).   \label{eq:comp}
\end{eqnarray}

If a graph $G$ has a connected component isomorphic to $K_2$ or at least two single vertices, there is no union vertex-distinguishing coloring of $G$. Hence we restrict ourselves in this paper to graphs with connected components of size at least $3$. As for the identifying coloring problem~\cite{these}, a lower bound is straightforwardly available for $\kiu(G)$:

\begin{prop}\label{prop:lowerBound}
Every graph $G=(V,E)$ where no connected component is isomorphic to a single vertex or an edge admits a union vertex-distinguishing coloring and satisfies
$$ \logV{V} \leq \kiu(G). $$
\end{prop}

\begin{proof}
Since $G$ admits no connected component isomorphic to a single vertex or an edge, there always exists a union vertex-distinguishing coloring of $G$. Indeed, for every pair of vertices $x,y$, the sets of edges incident to $x$ and $y$ are distinct. Thus, if one assigns different colors to each edge, all the pairs of vertices are incident to distinct subsets of colors.
In addition, since $\{1,\ldots ,k\}$ admits $2^{k}-1$ non empty subsets, a union vertex-distinguishing $k$-coloring must satisfy $|V|\leq 2^k-1$. Indeed, otherwise two vertices must receive identical non empty subsets of colors.
\end{proof}

A graph $G$ is {\em optimally colored} if $\kiu(G)=\logV{V}$. We will show that every graph can be colored using $\logV{V}$ or $\logV{V}+1$ or $\logV{V}+2$ colors. This result is somehow surprising since in almost all the known variants of identifying coloring, there exist graphs for which the gap between the lower bound and the minimum number of colors needed to identify all the vertices is arbitrarily large (see e.g.~\cite{these}). 
We can easily prove that there exist classes of graphs on which the optimal value can be obtained in any case (e.g. trees, paths, large enough cycles, complete binary trees...). In the case of complete graphs, the lower bound provided by Proposition~\ref{prop:lowerBound} is not always tight. Indeed, consider the complete graph $K_n$ of order $n=2^k-1$ for any integer $k>1$. If $\kiu(K_n)=\logV{V}$, then there exist two distinct vertices $u$ and $v$ such that $\id{u}=\{1\}$ and $\id{v}=\{2\}$. The only way to produce these codes is to color all the edges incident to $u$ with the singleton $\{1\}$, and all the edges incident to $v$ with $\{2\}$. Hence the edge $(u,v)$ would receive two different colors, a contradiction.
Harary proved in~\cite{Harary} that $\chi_S(K_n)=\lceil \log_2(n) \rceil +1$. Putting together this result with (\ref{eq:comp}), it ensures that we have $\kiu(K_n)\in\{\lceil \log_2(n) \rceil,\lceil \log_2(n) \rceil+1\}$.
We were not able to find a graph for which $\logV{V}+2$ colors are needed. Determining if there exist, or not, graphs that can reach this value remains open.
Vizing's theorem ensures that the minimum number of colors of an edge coloring of any graph of maximum degree $\Delta$ is either $\Delta$ (which is the immediate lower bound) or $\Delta+1$. Our result is similar since all the union vertex-distinguishing number is either the optimal value or the optimal value plus one or the optimal value plus two.

The paper is organized as follows. In Section~\ref{sec:general}, we prove our main result that ensures that for every graph $G$ that admits a union vertex-distinguishing coloring, the value of $\kiu(G)$ is between $\logV{V}$ and $\logV{V}+2$. In Section~\ref{sec:particularclasses}, we prove that the lower bound is reached for paths, complete binary trees, and cycles of length $n$ with $n\neq 7$.

\section{An almost tight upper bound on $\kiu(G)$}\label{sec:general}

\begin{theorem}
\label{theorem:superiorBound}
Every graph $G=(V,E)$ where no connected component is isomorphic to a single vertex or an edge satisfies
$$ \kiu(G) \leq \logV{V}+2. $$
\end{theorem}

This section is devoted to prove Theorem~\ref{theorem:superiorBound}.
Let us first describe in a few words the structure of the proof. A graph $H$ is an \emph{edge-subgraph} of $G$ if $V(H)=V(G)$ and $E(H) \subseteq E(G)$. First we show in Lemma~\ref{lem:subgraphImpliesGraph} that, for every edge-subgraph $H$ of $G$, we have $\kiu(G) \leq \kiu(H)+1$.
Thus, if we can find an edge subgraph $H$ of $G$ such that $\kiu(H) \leq \logV{V}+1$, the conclusion immediately holds. Lemmas~\ref{lem:treeDecomposition} and~\ref{lem:stars} consist in extracting such a subgraph and proving that it admits a union vertex-distinguishing coloring with $\logV{V}+1$ colors.

\begin{lemma}
\label{lem:subgraphImpliesGraph}
 Let $G=(V,E)$ be a graph. For any edge-subgraph $H$ of $G$, we have $\kiu(G) \leq \kiu(H)+1$.
\end{lemma}
\begin{proof}
Let $\alpha$ be a $k$-coloring of $H$ that is union vertex-distinguishing. We denote by $\{ 1,\ldots, k \}$ the colors of $\alpha$ and let $k+1$ be a new color. Let us prove that $G$ admits a union vertex-distinguishing coloring with $k+1$ colors. Consider the coloring $\beta$ of $G$ where $\beta(e)=\alpha(e)$ if $e \in E(H)$ and $\beta(e)=\{ k+1 \}$ otherwise.
For every vertex $u$ in $V(G)$, we have $\idC{\beta}{u}= \idC{\alpha}{u}$ or $\idC{\alpha}{u} \cup \{k+1 \}$. Indeed, if $u$ is not an endpoint of an edge of $E(G) \setminus E(H)$, then $\idC{\beta}{u}= \idC{\alpha}{u}$. Otherwise, color $k+1$ also appears in  an edge incident to $u$ and then $\idC{\beta}{u} = \idC{\alpha}{u} \cup \{k+1 \}$.

Thus, for every pair $u,v$ of vertices, we have $\idC{\beta}{u} \cap \{ 1,\ldots,k \} = \idC{\alpha}{u} \neq \idC{\alpha}{v} = \idC{\beta}{v} \cap \{ 1,\ldots,k\}$. Thus $\idC{\beta}{u} \neq \idC{\beta}{v}$ which completes the proof.
\end{proof}

The \emph{star} $S_{1,k}$ is a graph on $k+1$ vertices with $k$ vertices (called \emph{leaves}) of degree one all connected to the $(k+1)$-th vertex (called the \emph{center}). A star is \emph{non-trivial} if it admits at least two leaves.
A graph $H$ is a \emph{$1$-subdivision} of $G$ if $H$ can be obtained from $G$ by subdividing each edge at most one time.
A \emph{$1$-star} is a $1$-subdivision of a non-trivial star.

\begin{lemma}
\label{lem:treeDecomposition}
Any graph with no connected component isomorphic to a single vertex or an edge admits a disjoint union of $1$-stars as an edge-subgraph.
\end{lemma}
\begin{proof}
 Assume by contradiction that there exists a graph $G$ with no connected component isomorphic to a single vertex or an edge that does not admit a $1$-star as an edge-subgraph. Amongst all these counter-examples, choose $G$ that lexicographically minimizes its number of vertices and then its number of edges. Note that the minimality of $G$ ensures that $G$ is connected and that $G$ has at least $3$ vertices.

 For every edge $e$, if $G \setminus e$ admits as an edge-subgraph a disjoint union of $1$-stars as an edge-subgraph, then $G$ also does since $G \setminus e$ is an edge subgraph of $G$. So $G \setminus e$ has no decomposition into disjoint $1$-stars. Thus, by minimality of $G$, for every edge $e$, the graph $G \setminus e$ has at least one connected component that is reduced to a single vertex or an edge.

 Let $e=(u,v)$ be an edge where $v$ is a vertex of maximum degree. Note that $d(v) \geq 2$.
 First assume that $d(v) \geq 3$. In the graph $G'=G \setminus e$, the component of $u$ or the component of $v$ have size at most $2$. Since $d(v) \geq 2$ in $G'$, the component of $u$ in $G'$ is a single vertex or an edge.
 For every neighbor $w$ of $v$ the component of $w$ in $G \setminus e'$, where $e'=(v,w)$, is either is single vertex or an edge. Indeed the component of $v$ contains both $u$ and $v$ and another neighbor of $v$. Thus, the component of $v$ in $G \setminus e$ is a $1$-star with center $v$. And hence, so is the component of $v$ in $G$.

 Assume now that $d(v)=2$. By connectivity of $G$, $G$ is a path or a cycle with at least three vertices. Then $G$ admits as an edge-subgraph a disjoint union of $P_3$, $P_4$ and $P_5$ which are all $1$-stars, a contradiction.

 \end{proof}

\begin{lemma}
\label{lem:stars}
Any $1$-star can be optimally colored.
\end{lemma}
\begin{proof}
Let $S$ be a $1$-star. One can easily check that $P_3$ can be colored using $2$ colors (Theorem~\ref{theorem:paths} actually ensures that all the paths can be optimally colored). From now on we assume that $S \neq P_3$. Let us denote by $u$ the center of $S$ and by $n_1$ and $n_2$ the number of vertices at distance respectively $1$ and $2$ from $u$. Note that $n:=|V|=1+n_1+n_2$ and that $n_1 \geq n_2$ since each vertex incident to $u$ is incident to at most one vertex in the second neighborhood of $u$. Let $Y$ be the set of non-neighbors of $u$ distinct from $u$ and $X$ be the set of vertices in $N(u)$ incident to a vertex of $Y$. By definition of $1$-stars, $|X|=|Y|$ and $(X,Y)$ induces a perfect matching. Finally, let $Z =N(u) \setminus X$.

Let us denote by $k$ the integer $\logN$. Since $X$ and $Y$ are disjoint and neither of them contain $u$, we have $|X|\leq\frac{n-1}{2}\leq2^{k-1}-1$.
Assume first that $|X|\neq 2^{k-1}-1$. Let $\alpha$ be a coloring satisfying the following properties:
\begin{enumerate}[(i)]
 \item For every vertex $x \in X$, $\alpha(u,x)$ is a strict subset of $\{ 1, \ldots, k \}$ of size at least $2$ that contains color $k$ such that $\alpha(u,x) \neq \alpha(u,x')$ if $x \neq x'$. Moreover, if $X$ contains at least two vertices, then $\cup_{x \in X} \alpha(u,x) = \{ 1,\ldots, k\}$. Such sets necessarily exist since there are $2^{k-1}-2$ strict subsets of size at least $2$ of $\{1,\ldots,k\}$ containing $k$ and $|X|\leq 2^{k-1}-2$.
 \item Every edge $(x,y)$ between $x$ of $X$ and $y$ of $Y$ satisfies $\alpha(x,y) = \alpha(u,x) \setminus k$.
 \item Every edge $(u,z)$ with $z \in Z$, $\alpha(u,z)$ is a strict non-empty subset of $\{ 1, \ldots, k \}$ that has not been assigned yet to an edge. Moreover, if $Z$ is non-empty, the set $\{1,\ldots, k-1 \}$ must be one of these sets. If $X$ is empty then the label of at least one edge must also contain $k$. This is possible since there are at most $2^k-2$ edges in $G$ and $2^k-2$ strict non-empty subsets of $\{1,\ldots,k\}$.
\end{enumerate}



Note that for each vertex $v\neq u$, $\id{v}$ corresponds to a unique $\alpha(v,w)$ for a neighbor $w$ of $v$. This is clear for vertices of $Y\cup Z$ that have degree one. If $v\in X$ and $w$ is its neighbor in $Y$, we have by construction $\alpha(v,w)\subset \alpha(v,u)$, thus $\id{v}=\alpha(v,u)$. Moreover $\id{u}= \{ 1,\ldots, k \}$. Indeed, the degree of $u$ is at least $2$. If the size of $X$ is at least two then, by construction $\cup_{x \in X} \alpha(u,x) = \{ 1,\ldots, k\}$ and thus $\id{u} =\{1,\ldots,k \}$. If the size of $X$ is one, then $u$ has another neighbor in $Z$. Since $\alpha(u,x)$ contains $k$ and one edge between $u$ and $Z$ is labeled $\{1,\ldots, k-1 \}$, we have $\id{u}= \{1,\ldots,k\}$. Finally, if $X$ is empty, then one edge is labeled by $\{1,\ldots,k-1\}$ and one label contains $k$. Thus $\id{u} = \{1,\ldots,k\}$.

Since all the edges get strict distinct subsets of $\{1,...,k\}$, $\alpha$ is union vertex-distinguishing. 

Consider now the case $|X|=2^{k-1}-1$. It means that $Z$ is empty. Let $x_1,x_2$ two vertices of $X$ and $y_1$, $y_2$ their neighbors in $Y$. We construct $\alpha$ as follows:
\begin{enumerate}[(i)]
\item $\alpha(x_1,u)=\{1,\ldots,k-2,k\}$;
\item $\alpha(x_1,y_1)=\{k\}$;
\item $\alpha(x_2,u)=\{1,\ldots,k-1\}$;
\item $\alpha(x_2,y_2)=\{1,\ldots,k-2\}$;
\item For every other vertex $x \in X$, $\alpha(u,x)$ is a strict subset of $\{ 1, \ldots, k \}$ of size at least $2$ that contains color $k$ and that has not been assigned yet. Such sets necessarily exist since there are $2^{k-1}-3$ strict subsets of size at least $2$ of $\{1,\ldots,k\}$ containing $k$ that have not been assigned yet and there remain $2^{k-1}-3$ vertices in $X$.
 \item Every edge $(x,y)$ between $x$ of $X\setminus\{x_1,x_2\}$ and $y$ of $Y\setminus\{y_1,y_2\}$ satisfies $\alpha(x,y) = \alpha(u,x) \setminus k$.
\end{enumerate}

As before, for each vertex $v\neq u$, $\id{v}$ corresponds to a unique $\alpha(v,w)$ for a neighbor $w$ of $v$ and $\id{u}= \{ 1,\ldots, k \}$. Since all the edges get again strict distinct subsets of $\{1,...,k\}$, $\alpha$ is union vertex-distinguishing.

\end{proof}

According to Lemma~\ref{lem:treeDecomposition}, we know that every graph with no connected component isomorphic to a vertex or an edge admits a disjoint union of $1$-stars as an edge-subgraph. Since adding edges to a graph only costs one color, it remains to prove that the disjoint union of $1$-stars admits a union vertex-distinguishing edge coloring using the optimal number of colors plus one in order to prove Theorem~\ref{theorem:superiorBound}.

\begin{lemma}\label{lem:disjointoptimal}
A disjoint union of graphs that can be separately optimally colored can be colored together using at most the optimal number of colors plus one.
\end{lemma}
\begin{proof}
Let us first give a few general definitions. Given two graphs $H_1$ and $H_2$, we denote by $H_1 \cup H_2$ the disjoint union of $H_1$ and $H_2$.
A graph $K$ is a $k$-graph if its number of vertices is between $2^k$ and $2^{k+1}-1$. Note that if $H_1$ and $H_2$ are two $k$-graphs then $H_1 \cup H_2$ is a $(k+1)$-graph.

\begin{claim}\label{clm:unionoptimal}
 Let $H_1$ and $H_2$ be two $k$-graphs that can be optimally colored. Then $H_1 \cup H_2$ can be optimally colored.
\end{claim}
\begin{proof}
Indeed, let $\alpha$ (respectively $\beta$) be an optimal union vertex-distinguishing coloring of $H_1$ (resp. $H_2$) both using colors $\{1,\ldots,k\}$. Then consider the coloring $\gamma$ of $H_1 \cup H_2$ such that $\gamma(e)=\alpha(e)$ if $e$ is an edge of $H_1$ and $\gamma(e)=\beta(e) \cup \{ k+1 \}$ if $e$ is an edge of $H_2$. We claim that $\gamma$ is union vertex-distinguishing. Indeed two vertices $u,v$ of $H_1$ satisfies $\id{u} \neq \id{v}$ since $\alpha$ is distinguishing. Similarly, two vertices $u,v$ of $H_2$ satisfies $\id{u} \cap \{1,\ldots,k \} \neq \id{v} \cap \{1,\ldots,k\}$ since $\beta$ is distinguishing. Now let $u \in V(H_1)$ and $v \in V(H_2)$. Since any vertex of $H_2$ is incident to an edge of $H_2$, $k+1$ appears in $\id{v}$ while it does not appear in $\id{u}$. This completes the proof.

Thus $H_1 \cup H_2$ can be colored using $(k+1)$ colors. Since we noticed that if $H_1$ and $H_2$ are two $k$-graphs then $H_1 \cup H_2$ is a $(k+1)$-graph, this coloring is optimal.
\end{proof}

Using Claim~\ref{clm:unionoptimal}, let us prove the lemma.
Let $\mathcal{H}=\{ H_1,\ldots,H_\ell \}$ be $\ell$ graphs that can be optimally colored. Let us prove that their disjoint union can be colored using at most the optimal number of colors plus one.
We prove it by induction on the number of graphs in $\mathcal{H}$. The case $\ell=1$ is clear.

First assume that there exists an integer $k$ such that $\mathcal{H}$ contains two $k$-graphs, without loss of generality, we can assume that these graphs are $H_1$ and $H_2$. By Claim~\ref{clm:unionoptimal}, $H_1\cup H_2$ can be optimally colored. Thus $\mathcal{H'} = \mathcal{H} \cup (H_1 \cup H_2) \setminus \{ H_1,H_2 \}$ contains less graphs than $\mathcal{H}$ and all of them can be optimally colored. By induction their union, which also is the disjoint union of $ H_1,\ldots,H_\ell$ can be colored using at most the optimal number of colors plus one, which achieves the proof in that case.

Thus we can assume that for every $k$, there is at most one $k$-graph. Order the graphs of $\mathcal{H}$ in increasing size. Let us prove by induction on $i$ that if $H_i$ is a $k_i$-graph then the disjoint union of $H_1,H_2,\cdots,H_i$ can be colored using $k_i+1$ colors. The base case $i=1$ is trivial: $H_1$ can be colored using $k_1$ colors so it can be colored using $k_1+1$ colors.
Now, let us consider the inductive case. Let $H_1,\cdots,H_i$ be the graphs in $\mathcal{H}$.
By induction hypothesis there exists a coloring $\alpha$ using $k_{i-1}+1$ colors to color $H_1 \cup \cdots \cup H_{i-1}$. Moreover, since $H_i$ is a $k_i$-graph, there exists a $k_i$-coloring of $H_i$. Note that $k_i \geq k_{i-1}+1$ since there is a unique $k_i$-graph. Define $\gamma$ as $\alpha(e)$ if $e$ is an edge of $H_j$ with $j<i$ or $\beta(e) \cup \{ k_i+1 \}$ if $e \in E(H_i)$. The coloring $\gamma$ is a $(k_i+1)$-coloring of $H_1 \cup \cdots \cup H_i$. As in Claim~\ref{clm:unionoptimal}, we can show that it is vertex-distinguishing. Indeed two vertices $u,v$ of $H_1\cup \cdots \cup H_{i-1}$ satisfies $\id{u} \neq \id{v}$ since $\alpha$ is vertex-distinguishing. Similarly, two vertices $u,v$ of $H_i$ satisfies $id(u) \cap \{1,\ldots,k_i \} \neq \id{v} \cap \{1,\ldots,k_i\}$ since $\beta$ is vertex-distinguishing. Now let $u \in V(H_1 \cup \cdots \cup H_{i-1})$ and $v \in V(H_i)$. Since any vertex of $H_2$ is incident to an edge of $H_2$, $k_i+1$ appears in $\id{v}$ while it does not appear in $\id{u}$. Thus all the vertices can be identified.

Finally, $H = H_1 \cup \cdots \cup H_\ell$ can be colored using $k_\ell+1$ colors. By definition of $k_\ell$, we have
\[ k_\ell \leq \logV{V(H_\ell)} \leq \logV{V(H)}.\]
Thus the graph $H$ can be colored with the optimal number of colors plus one.
\end{proof}

Let us finally combine all these lemmas to prove Theorem~\ref{theorem:superiorBound}. Let $G$ be a graph that does not contain any connected component of size at most $2$. According to Lemma~\ref{lem:treeDecomposition}, $G$ admits a graph $H$ as an edge-subgraph that is a disjoint union of $1$-stars. Lemma~\ref{lem:stars} ensures that each of these $1$-stars can be optimally colored. Thus, the graph $H$ is the disjoint union of graphs that can be optimally colored. Lemma~\ref{lem:disjointoptimal} ensures that the whole graph $H$ can be colored using at most $\logV{V(H)}+1$ colors. Finally, since $H$ is an edge-subgraph of $G$, Lemma~\ref{lem:subgraphImpliesGraph} ensures that $G$ can be colored using at most $\logV{V(H)}+2$ colors. Thus Theorem~\ref{theorem:superiorBound} holds.

\section{Exact values for several classes of graphs}\label{sec:particularclasses}

In this section, we show that several sparse classes of graphs can be optimally colored. In particular, we show that paths, cycles (of large enough length) and complete binary trees can be optimally colored. All these results can be seen as a first step in order to improve the result of Theorem~\ref{theorem:superiorBound}. Indeed, if we can prove that every tree (resp. forest) can be colored with the optimal number of colors, then using Lemma~\ref{lem:subgraphImpliesGraph}, we can prove that any connected (resp. any) graph can be colored with at most $\logV{V(G)}$ colors since every connected graph admits a spanning tree (resp. forest). 

In Section~\ref{sec:path}, we prove that paths can be optimally colored, in Section~\ref{sec:cycles}, we focus on cycles, and finally, we consider the case of complete binary trees in Section~\ref{sec:binary}

\subsection{Paths}\label{sec:path}

This section is devoted to prove that paths can be optimally colored. We actually prove a slightly larger statement that will be useful for the cycle case (see Section \ref{sec:cycles}).

\begin{theorem}
\label{theorem:paths}
For $n \geq 3$, the path $P_n$ can be optimally colored. Moreover, there exists an optimal union vertex-distinguishing $m$-coloring of $P_{n}$ such that:
\begin{enumerate}[(i)]
\item $\id{u_1} = \{ 1 \}$;
\item $\id{u_{n}}=\{ m \}$;
\item the only vertex that can satisfy $\id{u_j}=\{1,m\}$ is $u_{n-1}$;
\end{enumerate}
where $u_1$,...,$u_n$ are the vertices of $P_n$.
\end{theorem}

\begin{proof}

We prove the result by induction on $n$. The case $n=3$ is given by Figure \ref{fig:P3}.

\begin{figure}[!h]
\centering
\begin{tikzpicture}
  \foreach \I in {0,1,2}
  \node[noeud] at (\I,0) {};
  \draw (0,0)--(2,0);
	\draw (0.5,0.3) node {$\{1\}$};
	\draw (1.5,0.3) node {$\{2\}$};
\end{tikzpicture}
\caption{A union vertex-distinguishing coloring of $P_3$.}
\label{fig:P3}
\end{figure}
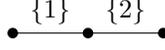

Let $n=2^k+\ell$ with $0\leq \ell<2^k$ and $k\geq 2$. By induction, there exists a union vertex-distinguishing coloring $\alpha_k$ of $P_{2^k-1}$ using $k$ colors and satisfying Conditions (i) to (iii). Using $\alpha_k$, we will construct a union vertex-distinguishing coloring $\beta$ of $P_n$ with $k+1$ colors and satisfying Conditions (i) to (iii). The vertices of $P_{2^k-1}$ are denoted $v_1$,...,$v_{2^k-1}$. Conditions (i) and (ii) will be trivially satisfied so we do not explicitly check them.

\

\noindent\textbf{Case 1:} $\ell=0$ ($n=2^k$)

We define the following coloring $\beta$ of $P_{n}$:
\begin{itemize}
\item $\beta(u_i,u_{i+1}) = \alpha_k(v_i,v_{i+1})$ for $1 \leq i \leq 2^{k}-2$
\item $\beta(u_{2^k-1},u_{2^k}) = \{ k+1 \}$
\end{itemize}

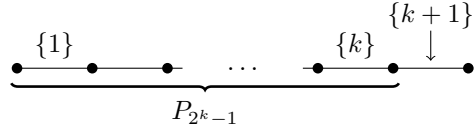
\begin{figure}[!h]
\centering
\begin{tikzpicture}[decoration={brace}]

\foreach \I in {0,1,2,4,5,6}
\node[noeud] at (\I,0) {};

\draw (3,0) node {$\ldots$};

\draw (0.5,0.3) node {$\{ 1 \}$};
\draw (4.5,0.3) node {$\{ k \}$};
\draw (5.5,0.7) node {$\{ k+1 \}$};
\draw [->] (5.5,0.5) -- (5.5,0.1);

\draw (0,0) -- (2.2,0) (3.8,0) -- (6,0);

\node (corner1) at (-0.2,-0.2) {};
\node (corner2) at (5.2,-0.2) {};
\node (pn1) at (2.5,-0.6) {$P_{2^k-1}$};
\draw [decorate,line width=1pt] (corner2) -- (corner1);

	
\end{tikzpicture}
\caption{The coloring of $P_{2^k}$ using the coloring $\alpha_{2^k-1}$.}
\label{fig:P2K}
\end{figure}
The construction is illustrated on Figure \ref{fig:P2K}.
The coloring $\beta$ is union vertex-distinguishing. Indeed the vertices $u_1,\ldots,u_{2^k-2}$ are pairwise distinguished by definition of $\alpha_k$. Moreover, we have $\idC{\beta}{u_{2^k-1}} = \{ k,k+1 \}$ and $\idC{\beta}{u_{2^k}} = \{ k+1 \}$. Thus these vertices can be distinguished from the others since for all $i \leq 2^k-2$, $\idC{\beta}{u_i}$ does not contain $k+1$. Finally, Condition (iii) is satisfied since there is no vertex with $\idC{\beta}{u_i}=\{1,k+1\}$. \medskip

\noindent\textbf{Case 2:} $\ell=1$

We define the following coloring $\beta$ of $P_{n}$:
\begin{itemize}
\item $\beta(u_i,u_{i+1}) = \alpha_{k}(v_i,v_{i+1})$ for $1 \leq i \leq 2^{k}-2$
\item $\beta(u_{2^k-1},u_{2^k}) = \{ k \}$
\item $\beta(u_{2^k},u_{2^k+1}) = \{ k+1 \}$
\end{itemize}

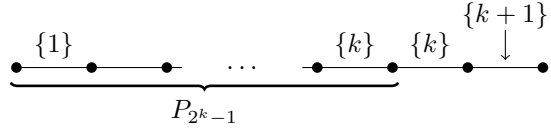
\begin{figure}[!h]
\centering
\begin{tikzpicture}[decoration={brace}]
\foreach \I in {0,1,2,4,5,6,7}
\node[noeud] at (\I,0) {};

	\node (corner1) at (-0.2,-0.2) {};
	\node (corner2) at (5.2,-0.2) {};
	\node (pn1) at (2.5,-0.6) {$P_{2^k-1}$};
	\draw [decorate,line width=1pt] (corner2) -- (corner1);
	\draw (0.5,0.3) node {$\{ 1 \}$};
	\draw (3,0) node {$\ldots$};
	\draw (4.5,0.3) node {$\{ k \}$};
	\draw (0,0) -- (2.2,0);
	\draw (3.8,0) -- (5,0);
	
	\draw (5.5,0.3) node {$\{ k \}$};
	\draw (6.5,0.7) node {$\{ k+1 \}$};
	\draw [->] (6.5,0.5) -- (6.5,0.1);
	\draw (5,0) -- (7,0);
\end{tikzpicture}
\caption{The coloring of $P_{2^k+1}$ using the coloring $\alpha_{k}$.}
\label{fig:P2K1}
\end{figure}
This construction is illustrated on Figure \ref{fig:P2K1}.
The coloring $\beta$ is union vertex-distinguishing. Indeed, the vertices $u_1,\ldots,u_{2^k-1}$ are pairwise distinguished since $\alpha_{k}$ is union vertex-distinguishing (note that, unlike the previous case, we have $\idC{\beta}{u_{2^k-1}} = \{ k \} = \idC{\alpha_{k}}{v_{2^k-1}}$). Moreover, we have $\idC{\beta}{u_{2^k+1}} = \{ k+1 \}$. Thus $u_{2^k+1}$ can be distinguished from the other vertices since  for all $1\leq i \leq 2^k-1$, $\idC{\beta}{u_i}$ does not contain $k+1$. Finally, there is no vertex with $\idC{\beta}{u_i}=\{1,k+1\}$.
\medskip

\noindent\textbf{Case 3:} $\ell=2$

We define the following coloring $\beta$ of $P_{n}$:
\begin{itemize}
\item $\beta(u_i,u_{i+1}) = \alpha_{k}(v_i,v_{i+1})$ for $1 \leq i \leq 2^{k}-2$
\item $\beta(u_{2^k-1},u_{2^k}) = \{ k \}$
\item $\beta(u_{2^k},u_{2^k+1}) = \{ 1,k+1 \}$
\item $\beta(u_{2^k+1},u_{2^k+2}) = \{ k+1 \}$
\end{itemize}

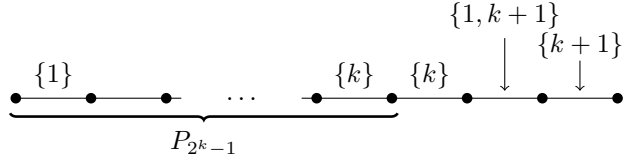
\begin{figure}[!h]
\centering
\begin{tikzpicture}[decoration={brace}]
\foreach \I in {0,1,2,4,5,6,7,8}
\node[noeud] at (\I,0) {};

	\node (corner1) at (-0.2,-0.2) {};
	\node (corner2) at (5.2,-0.2) {};
	\node (pn1) at (2.5,-0.6) {$P_{2^k-1}$};
	\draw [decorate,line width=1pt] (corner2) -- (corner1);
	\draw (0.5,0.3) node {$\{ 1 \}$};
	\draw (3,0) node {$\ldots$};
	\draw (4.5,0.3) node {$\{ k \}$};
	\draw (0,0) -- (2.2,0);
	\draw (3.8,0) -- (5,0);
	
	\draw (5.5,0.3) node {$\{ k \}$};
	\draw (6.5,1.1) node {$\{ 1,k+1 \}$};
	\draw (7.5,0.7) node {$\{ k+1 \}$};
	\draw [->] (6.5,0.8) -- (6.5,0.1);
        \draw [->] (7.5,0.5) -- (7.5,0.1);
	\draw (5,0) -- (8,0);
\end{tikzpicture}
\caption{The coloring of $P_{2^k+2}$ using the coloring $\alpha_{k}$.}
\label{fig:P2K2}
\end{figure}
This construction is illustrated on Figure \ref{fig:P2K2}.
The coloring $\beta$ is union vertex-distinguishing. Indeed, the vertices $u_1,\ldots,u_{2^k-1}$ are pairwise distinguished since $\alpha_{k}$ is union vertex-distinguishing. Moreover, we have $\idC{\beta}{u_{2^k}} = \{ 1,k,k+1 \}$, $\idC{\beta}{u_{2^k+1}} = \{ 1,k+1 \}$ and  $\idC{\beta}{u_{2^k+2}} = \{ k+1 \}$. Thus these vertices are distinguished from the other vertices since for all $1\leq i \leq 2^k-1$, $\idC{\beta}{u_i}$ does not contain $k+1$. Finally, the vertex with $\idC{\beta}{u_i}=\{1,k+1\}$ is $u_{n-1}$.
\medskip

\noindent\textbf{Case 4:} $3 \leq \ell \leq 2^{k}-1$

We denote the vertices of $P_\ell$ by $w_1,\ldots,w_\ell$. By induction, there exists a union vertex-distinguishing coloring $\alpha_\ell$ of $P_\ell$ satisfying satisfying Conditions (i) to (iii). Let $m$ be the number of colors in $\alpha_\ell$.
We define the following coloring $\beta$ of $P_{n}$:
\begin{itemize}
\item $\beta(u_i,u_{i+1}) = \alpha_{k}(v_i,v_{i+1})$ for $1 \leq i \leq 2^{k}-2$
\item $\beta(u_{2^k-1},u_{2^k}) = \{ k \}$
\item $\beta(u_{2^k+i},u_{2^k+i+1}) = \alpha_\ell(w_{\ell-i},w_{\ell-i-1}) \cup \{ k+1 \}$ for $0 \leq i \leq \ell-2$
\item $\beta(u_{n-2},u_{n-1}) = \{ 1,k+1 \}$
\item $\beta(u_{n-1},u_{n}) = \{ k+1 \}$
\end{itemize}

\begin{figure}[!h]
\centering
\begin{tikzpicture}[decoration={brace}]
  \foreach \I in {0,1,4,5,6,7,9,10,11}
\node[noeud](\I) at (\I,0) {};

	\node (corner1) at (-0.2,-0.2) {};
	\node (corner2) at (5.2,-0.2) {};

	\draw [decorate,line width=1pt] (corner2) -- node[below] {$P_{2^k-1}$} (corner1);
	\draw (0.5,0.3) node {$\{ 1 \}$};
	\draw (2.5,0) node {$\ldots$};
	\draw (4.5,0.3) node {$\{ k \}$};
	\draw (0,0) -- (1.2,0);
	\draw (3.8,0) -- (5,0);
	
	\draw (5.5,0.3) node {$\{ k \}$};
	\draw (5,0) -- (7.2,0);
	
	\node (corner4) at (10.2,-0.2) {};
	\node (corner3) at (5.8,-0.2) {};
        \draw [decorate,line width=1pt] (corner4) --  
node[text centered, text width=4cm,below] {
          reversed $P_\ell$ with color $k+1$ added to every edge}
        (corner3);

	\draw (8,0) node {$\ldots$};
	\draw (6.5,1.2) node {$\{ m,k+1 \}$};
	\draw [->] (6.5,1) -- (6.5,0.1);
	
        \draw (9.5,1.2) node {$\{ 1,k+1 \}$};
	\draw [->] (9.5,1) -- (9.5,0.1);
	\draw (8.8,0) -- (10,0);
	
	\draw (10.5,0.7) node {$\{ k+1 \}$};
	\draw [->] (10.5,0.5) -- (10.5,0.1);
	\draw (11,0) -- (10,0);
\end{tikzpicture}
\caption{The coloring of $P_{2^k+\ell}$ using the ones of $P_{2^k-1}$ and $P_\ell$.}
\label{fig:P2KJ}
\end{figure}
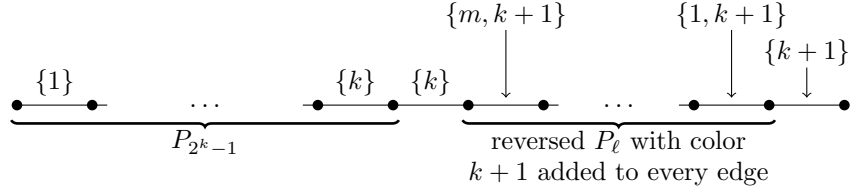
This construction of the coloring is illustrated on Figure \ref{fig:P2KJ}.
The vertices $u_1,\ldots,u_{2^k-1}$ are pairwise distinguished by definition of $\alpha_{k}$. The same holds for the vertices of $u_{2^k},\ldots,u_{n-1}$ by definition of $\alpha_\ell$. Indeed, $\idC{\beta}{u_{2^k+i}}\cap \{1,...,m\}=\idC{\alpha_\ell}{w_{\ell-i}}$ for $i\in {0,...\ell-1}$.

Besides, for $0 \leq i \leq 2^k-1$ and $2^k \leq i' \leq n$, $\idC{\beta}{u_i} \neq \idC{\beta}{u_{i'}}$, since $\{ k+1 \} \in \idC{\beta}{u_{i'}}$ and $\{ k+1 \} \notin \idC{\beta}{u_i}$. The vertex $u_{n}$ is also identified since no other vertex can have the label $\{k+1\}$. Finally, the coloring is optimal since it uses $k+1$ colors on a path of length $2^k+\ell$ with $0 \leq \ell \leq 2^k$.

\medskip
\end{proof}

\subsection{Cycles}\label{sec:cycles}

In this section, we prove that cycles of length different from 3 or 7 can be optimally colored. First note that every cycle can be colored with $\logV{G}$ or $\logV{G}+1$ colors. Indeed, the path $P_n$ is an edge-subgraph of the cycle $C_n$. By Theorem~\ref{theorem:paths}, the path $P_n$ can be optimally colored. Thus by Lemma~\ref{lem:subgraphImpliesGraph}, the graph $C_n$ can be colored with $\logV{G}+1$ colors. The remaining of this section is devoted to show that the plus one fact can be dropped in all but two cases.

\begin{lemma}
\label{lem:c3}
$\kiu(C_3)=3$ and $\kiu(C_7)=4$
\end{lemma}

\begin{proof}
We mentioned in the introduction that no clique of size $2^k-1$ can be colored with only $k$ colors. Since $C_3=K_3$, the result holds for $C_3$.

We now concentrate on the graph $C_7$. Let $u_1,\ldots,u_7$ be the seven vertices in $C_7$. We know by Theorem \ref{theorem:superiorBound} that at least three colors are required. We will prove that three colors are not enough.
We proceed by contradiction: assume $C_7$ is distinguished with a coloring $\alpha$ using only three colors. Thus, there are three vertices $v_1,v_2$ and $v_3$ such that $\id{v_i}=\{i\}$. Since all the edges incident to $v_i$ are labeled with $i$, $v_1,v_2$ and $v_3$ are pairwise non incident. Without loss of generality, we can assume that $\id{u_1}=\{1\}$, $\id{u_3}=\{2\}$ and $\id{u_5}=\{3\}$. Thus we have $\id{u_2}=\{1,2\}$ and $\id{u_4}=\{2,3\}$.
Now, we need to have the labels $\{1,3\}$ and $\{1,2,3\}$ on the vertices $u_6$ and $u_7$. However, since $2\not\in \alpha(u_7,u_1),\alpha(u_5,u_6)$, we have $2\in\alpha(u_6,u_7)$, which is a contradiction with the fact that either $u_6$ or $u_7$ has the label $\{1,3\}$.
Thus we need at least four colors to distinguish $C_7$.

%
\end{proof}

\begin{theorem}
\label{theorem:cycles}
For $n \geq 4$, $n \neq 7$, $C_n$ can be optimally colored.
\end{theorem}

We first prove the case $n\neq 2^k-1$.

\begin{lemma}
\label{lem:cycles}
Let $n \geq 4$ such that $n \neq 2^k-1$ for any $k$. Then $\kiu(C_n)= \logN$.
\end{lemma}

\begin{proof}
Let $k\geq 2$ and $n$ such that $2^k\leq n < 2^{k+1}-1$.
We denote the vertices of $C_n$ as $u_1,\ldots,u_n$.
We denote by $v_1,\ldots,v_{n+1}$ the vertices of $P_{n+1}$. Since $\lceil \log_2(n+2) \rceil =k+1$, Theorem~\ref{theorem:paths} ensures that there exists a union vertex-distinguishing coloring $\alpha$ of $P_{n+1}$ using $k+1$ colors with $\idC{\alpha}{v_1}=\{1\}$, $\idC{\alpha}{v_{n+1}}=\{k+1\}$ and where the only vertex of $P_{n+1}$ that can be identify by $\{1,k+1\}$ is $v_n$. 

We define the following coloring $\beta$ of $C_n$ using $k+1=\logN$ colors:
\begin{itemize}
\item $\beta(u_i,u_{i+1})=\alpha(v_i,v_{i+1})$ for $1 \leq i \leq n-1$
\item $\beta(u_1,u_n) = \left\{
\begin{array}{ll}
\{1\} & \mbox{if } 1 \in \idC{\alpha}{v_n} \\
\{k+1\} & \mbox{otherwise}
\end{array}
\right.$
\end{itemize}

This construction is shown on Figure \ref{fig:CN}.

\begin{figure}[!h]
\centering
\begin{tikzpicture}[decoration={brace}]
  \foreach \I in {0,1,3,4,5}
  \node[noeud] at (\I,0) {};
	\node (corner1) at (-0.2,0.6) {};
	\node (corner2) at (5.2,0.6) {};
	\node (pn1) at (2.5,1) {$P_{n+1}$};
	\draw [decorate,line width=1pt] (corner1) -- (corner2);
	\draw (0.5,0.3) node {$\{ 1 \}$};
	\draw (2,0) node {$\ldots$};
	\draw (3.5,0.3) node {$E$};
	\draw (0,0) -- (1.2,0);
	\draw (2.8,0) -- (5,0);
	
	\draw (0,0) to[out=-20,in=200] node[below, text centered, text width=3cm] {$\{ 1 \}$ if $1 \in E$, $\{ k+1 \}$ otherwise} (4,0);
\end{tikzpicture}
\caption{The construction of $C_n$ from $P_{n+1}$.}
\label{fig:CN}
\end{figure}
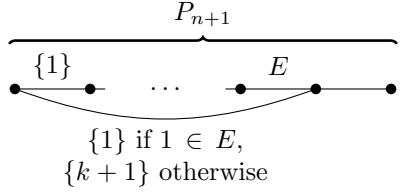

We now prove that $\beta$ is union vertex-distinguishing.
For every vertex $u_i$ with $2 \leq i \leq n-1$, $\idC{\beta}{u_i} = \idC{\alpha}{v_i}$. These vertices are distinguished since the corresponding vertices are pairwise distinguished in $\alpha$. Thus, we only need to study $u_1$ and $u_n$.
Since $\alpha$ satisfies the conclusion of Theorem~\ref{theorem:paths}, if a vertex is identified by $\{ 1,k+1\}$ in $P_{n+1}$, it must be $v_n$. In this case, we have $\idC{\beta}{u_1} = \idC{\alpha}{v_1}=\{1\}$ and $\idC{\beta}{u_n} = \idC{\alpha}{v_n}=\{1,k+1\}$, and those vertices are distinguished. Otherwise, we have $\idC{\beta}{u_n} = \idC{\alpha}{v_n}$ and $\idC{\beta}{u_1} = \{ 1,k+1 \} \neq \idC{\alpha}{v_i}$ (for $2 \leq i \leq n$) and those vertices are distinguished.
\end{proof}

To conclude the proof of Theorem~\ref{theorem:cycles}, we now deal with the case $n=2^k-1$.
\begin{lemma}
\label{lem:cycles2k-1}
Let $k \geq 4$ and $n=2^k-1$. Then $\kiu(C_{2^k-1})= k$.
\end{lemma}

\begin{proof}
We denote the vertices of $C_n$ by $u_1$,...,$u_n$.
We prove by induction on $k\geq 4$ that  there is a union vertex-distinguishing coloring $\alpha_k$ of $C_{2^k-1}$ with $k$ colors such that $\idC{\alpha_k}{u_1}=\{1\}$ and $1\in \alpha_k(u_2,u_3)$. 

The base case is $k=4$ (giving us $C_{15}$), that can be optimally colored as shown on Figure \ref{fig:C15}.

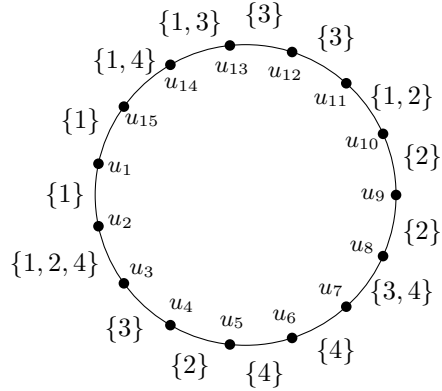
\begin{figure}[!h]
  \centering
  \begin{tikzpicture}
    \foreach \I in {1,...,15}
    {\node[noeud] at (24*\I-24*8:2) {};
      \draw (24*\I-24*9:1.7) node {\small $u_{\I}$};
    }
    \draw (0,0) circle (2);

	\draw (12:2.4) node {$\{2\}$};
	\draw (32:2.45) node {$\{1,2\}$};
	\draw (60:2.4) node {$\{3\}$};
	\draw (84:2.4) node {$\{3\}$};
	\draw (108:2.4) node {$\{1,3\}$};
	\draw (132:2.4) node {$\{1,4\}$};
	\draw (156:2.4) node {$\{1\}$};
	\draw (180:2.4) node {$\{1\}$};
	\draw (200:2.7) node {$\{1,2,4\}$};
	\draw (228:2.4) node {$\{3\}$};
	\draw (252:2.4) node {$\{2\}$};
	\draw (276:2.4) node {$\{4\}$};
	\draw (300:2.4) node {$\{4\}$};
	\draw (328:2.45) node {$\{3,4\}$};
	\draw (348:2.4) node {$\{2\}$};
\end{tikzpicture}
\caption{A vertex-distinguishing coloring of $C_{15}$.}
\label{fig:C15}
\end{figure}

Let $k\geq 4$ and $n=2^k-1$. Assume there exists a union vertex-distinguishing $k$-coloring $\alpha_k$ for $C_{n}$ that satisfies $\idC{\alpha_k}{u_1}=\{1\}$ and $1\in \alpha_k(u_2,u_3)$. We will construct a union vertex-distinguishing $(k+1)$-coloring $\beta$ of $C_{2n+1}$ using $\alpha_k$.

First, we create $C'_n$, a copy of $C_n$: every vertex $u_i$ in $C_n$ has a copy $u'_i$ in $C'_n$. We define the coloring $\alpha'_k$ on $C'_n$ by $\alpha'_k(u'_i,u'_{i+1})=\alpha_k(u_i,u_{i+1}) \cup \{k+1\}$ for all $1 \leq i \leq n-1$ and $\alpha'_k(u'_1,u'_{n})=\alpha_k(u_1,u_{n}) \cup \{k+1\}$, meaning that $\alpha'_k$ is a union vertex-distinguishing coloring for $C'_n$. Thus, we have: $$\bigcup_{i=1}^n (\idC{\alpha_k}{u_i} \cup \idC{\alpha'_k}{u'_i}) = \mathcal{P}^*(\{1,\ldots,k+1\}) \setminus \{k+1\}$$ (where $\mathcal{P}^*(S)$ denotes the powerset of $S$ short of the empty set).

We now create the following graph $G$, isomorphic to $C_{2n+1}$:
\begin{itemize}
\item $V(G)=V(C_n) \cup V(C'_n) \cup \{v\}$
\item $E(G)=(E(C_n) \setminus \{(u_1,u_{n})\}) \cup (E(C'_n) \setminus \{(u'_1,u'_{n})\}) \cup \{(u_1,v),(u'_1,v),(u_{n},u'_{n})\}$
\end{itemize}

We define the following coloring $\beta$ on $G$, which uses $k+1$ colors:
\begin{itemize}
\item $\beta(u_i,u_{i+1}) = \alpha_k(u_i,u_{i+1})$ for $0 \leq i \leq n-1$,
\item $\beta(u'_i,u'_{i+1}) = \alpha'_k(u'_i,u'_{i+1})$ for $0 \leq i \leq n-1$,
\item $\beta(u'_{1},u'_2) = \{k+1\}$
\item $\beta(u_1,v) = \{1\}$
\item $\beta(u'_1,v) = \{k+1\}$
\item $\beta(u_{n},u'_{n}) = \{1\} \cup \alpha_k(u_{n-1},u_{n})$
\end{itemize}

This construction is shown on Figure \ref{fig:cycles}.

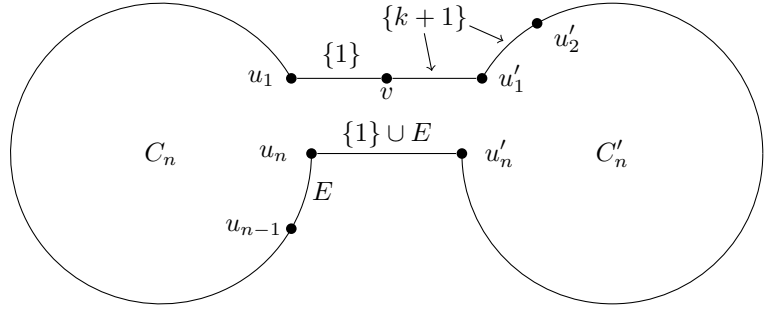
\begin{figure}[!h]
\centering
\begin{tikzpicture}
	\draw (1.732,1) arc (30:360:2);
	\draw (4,0) arc (180:360:2);
	\draw (8,0) arc (0:150:2);
        \node[noeud](un) at (2,0) {};
        \node[noeud](vn) at (4,0) {};
        \node[noeud](u1) at (1.732,1) {};
        \node[noeud](up1) at (4.268,1) {};
        \node[noeud](v) at (3,1) {};
        \node[noeud] at (5,1.732) {};
        \node[noeud] at (1.732,-1) {};

	\draw (u1)--(v)--(up1);
	\draw (un) -- (vn);
	
	\draw (0,0) node {$C_n$};
	\draw (6,0) node {$C'_n$};
	\draw (3,0.8) node {$v$};
	\draw (1.5,0) node {$u_{n}$};
	\draw (4.5,0) node {$u'_{n}$};
	\draw (1.332,1) node {$u_{1}$};
	\draw (4.668,1) node {$u'_1$};
	\draw (5.4,1.532) node {$u'_{2}$};
	\draw (1.232,-1) node {$u_{n-1}$};
	
	\draw (3,0.25) node {$\{1\} \cup E$};
	\draw (2.15,-0.5) node {$E$};
	\draw (2.4,1.3) node {$\{ 1 \}$};
	\draw (3.5,1.8) node {$\{ k+1 \}$};
	\draw [->] (3.5,1.6) -- (3.6,1.1);
	\draw [->] (4.1,1.7) -- (4.55,1.5);
\end{tikzpicture}
\caption{The construction of $G$ and $\beta$ from $C_n$, $C'_n$, $\alpha$ and $\alpha'$.}
\label{fig:cycles}
\end{figure}

We now prove that $\beta$ is union vertex-distinguishing. First, we can see that every vertex $u_i$ (resp. $u'_i$)  satisfies $\idC{\beta}{u_i}=\idC{\alpha_k}{u_i}$ (resp. $\idC{\beta}{u'_i}=\idC{\alpha'_k}{u'_i}$), except  $u'_1$. Indeed, this is clear for all the vertices but $u_1$, $u_n$, $u'_2$, $u'_n$ since the colors of their incident edges do not change. Since $\alpha_k(u_1,u_n)=\{1\}$ and $\alpha'_k(u'_1,u'_n)=\{1,k+1\}$ it is true for $u_1$, $u_n$ and $u'_n$. Finally, $1$ has been removed from the edge $(u'_1,u'_2)$ but by induction hypothesis, $1\in \alpha'_k(u'_2,u'_3)$ hence we also have $\idC{\beta}{u'_2}=\idC{\alpha'_k}{u'_2}$. Thus all these vertices are pairwise distinguished.

For the two last vertices, $\idC{\beta}{u'_1} = \{k+1\}$ does not appear in $\bigcup_{i=1}^n (\idC{\alpha_k}{u_i} \cup \idC{\alpha'_k}{u'_i})$ and $\idC{\beta}{v} = \{1,k+1\} = \idC{\alpha'}{u'_1}$. So these two vertices are distinguished from the other.

In conclusion, $\beta$ is a union vertex-distinguishing coloring of $C_{2^{k+1}-1}$ using $k+1$ colors. Furthermore, choosing for the first vertices $u_1$, $u_2$, $u_3$, $\ldots$ we have $\idC{\beta}{u_1}=\{1\}$ and still $1\in \beta(u_2,u_3)$. Hence the extra condition of the induction is still satisfied, which completes the proof.\end{proof}

\subsection{Complete Binary Trees}\label{sec:binary}

In this section we show that complete binary trees can be optimally colored. In what follows, all the complete binary trees will be rooted in their unique vertex of degree $2$. Recall that for such a tree $T$, the {\em height} of $T$ is the length (in terms of number of edges) of a path from the root to a leaf. Given a positive integer $h$, we will denote by $T_h$ and $r_h$, respectively, the complete binary tree of height $h$ and its root. Hence $T_{h+1}$ can be inductively built from two copies of $T_h$, say $T_h$ and $T'_h$, as follows:

\begin{eqnarray}
  V(T_{h+1})&=&V(T_h)\cup V(T'_h)\cup \{r_{h+1}\} \label{eq:thv}\\
  E(T_{h+1})&=&E(T_h)\cup E(T'_h)\cup \{(r_h,r_{h+1}),(r'_h,r_{h+1})\}\label{eq:the}
\end{eqnarray}

\begin{theorem}
\label{theorem:cbtrees}
For all $h\geq 1$, the complete binary tree $T_h$ can be optimally colored.
\end{theorem}

\begin{proof}
We proceed by induction on $h$. Given $h\geq 1$, our induction hypothesis claims there exists a union vertex-distinguishing $(h+1)$-coloring $\alpha_h$ of $T_h$ such that:
\begin{align}
  &\idC{\alpha_h}{r_h}=\{h,h+1\} \label{hyp:1} \\
  &\forall e\in E(T'_{h-1}),\; h+1\in \alpha_h(e). \label{hyp:2}
\end{align}

Note that such an $(h+1)$-coloring is an optimal coloring, since the complete binary tree of height $h$ has $2^{h+1}-1$ vertices.

This property holds for $h=1$, as shown by the coloring depicted on Figure~\ref{fig:P3} (with $r_h=u_2$).

Now assume the property holds for some $h\geq 1$, and consider $T_{h+1}$ as defined in \eqref{eq:thv} and \eqref{eq:the}. By induction hypothesis, there exists an $(h+1)$-coloring $\alpha_h$ of $T_h$ (and also of $T'_h$) satisfying Conditions~(\ref{hyp:1}) and (\ref{hyp:2}). We now define an $(h+2)$-coloring $\alpha_{h+1}$ of $T_{h+1}$ as follows:
\begin{flalign*}
&\forall e\in E(T_h),\; \alpha_{h+1}(e)=\alpha_h(e), \\
&\forall e\in E(T'_h),\; \alpha_{h+1}(e)=\alpha_h(e)\cup\{h+2\},\\
&\alpha_{h+1}(r_h,r_{h+1})=\{h+1\},\\
&\alpha_{h+1}(r'_h,r_{h+1})=\{h+1,h+2\}.
\end{flalign*}

Figure~\ref{fig:ABC} illustrates this coloring.
 \begin{figure}[!h]
 \begin{center}
 \includegraphics[scale=.3]{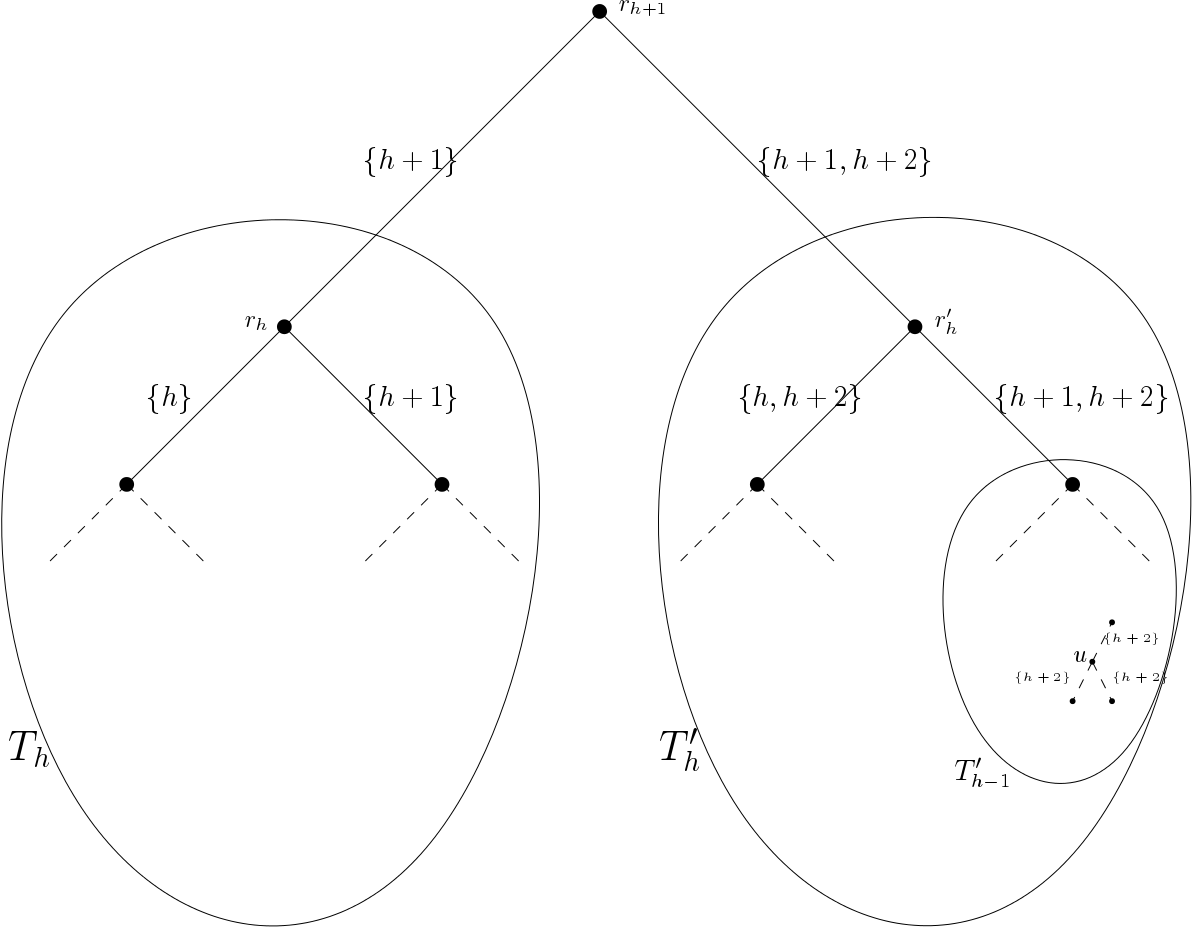}
 \caption{Optimal coloring of $T_{h+1}$ \label{fig:ABC}}
 \end{center}
 \end{figure}

With the above definition, we have:
\begin{flalign*}
&\forall v\in V(T_h)\setminus\{r_h\},\; \idC{\alpha_{h+1}}{v}=\idC{\alpha_{h}}{v},\\
&\forall v\in V(T'_h)\setminus\{r'_h\},\; \idC{\alpha_{h+1}}{v}=\idC{\alpha_{h}}{v}\cup\{h+2\},\\
&\idC{\alpha_{h+1}}{r_{h+1}}=\{h+1,h+2\}.\\
\end{flalign*}
Moreover, since $r_h$ and $r'_h$ both satisfy Condition~(\ref{hyp:1}), we also have that $\idC{\alpha_{h+1}}{r_h}=\idC{\alpha_{h}}{r_h}$ and $\idC{\alpha_{h+1}}{r'_h}=\idC{\alpha_{h}}{r_h}\cup\{h+2\}$. Hence, by induction hypothesis, all the vertices of $T_h$ are identified by $\alpha_{h+1}$ with all non-empty subsets of $\{1,\ldots,h+1\}$, and all the vertices of $T'_h$ are identified with all subsets of $\{1,\ldots,h+2\}$ containing $h+2$. In particular, there exists a unique vertex $u$ of $T'_h$ such that $\idC{\alpha_{h+1}}{u}=\{h+1,h+2\}$. Since $\idC{\alpha_{h}}{u}=\{h+1\}$ and since $T'_h$ satisfies (\ref{hyp:2}), it implies that $u\in T'_{h-1}$. Remark that $u$ cannot be adjacent to a leaf, otherwise they would not be identified by $\alpha_{h}$ as having the same code $\{h+1\}$.\\
Up to now, the coloring $\alpha_{h+1}$ is not vertex-distinguishing since $\idC{\alpha_{h+1}}{u}=\idC{\alpha_{h+1}}{r_{h+1}}=\{h+1,h+2\}$ and since $\{h+2\}$ has not been assigned yet to a vertex. For every edge $e$ incident to $u$, remove the value $h+1$ from the set $\alpha_{h+1}(e)$. Hence $\idC{\alpha_{h+1}}{u}=\{h+2\}$. In addition, since $u$ is in $T'_{h-1}$ and if it is not the root of $T'_{h-1}$, all the vertices adjacent to $u$ have their code unchanged since they are of degree $3$ and, by (\ref{hyp:2}), are incident to at least an edge with the color $h+1$. Otherwise, if $u$ is the root of $T'_{h-1}$, then $r'_h$ has still $h+1$ in its code since $\alpha_{h+1}(r'_h,r_{h+1})=\{h+1,h+2\}$ (see Figure~\ref{fig:ABC}). Therefore, the coloring $\alpha_{h+1}$ is vertex-distinguishing and Conditions (\ref{hyp:1}) and (\ref{hyp:2}) are satisfied for $T_{h+1}$, which completes the proof.
\end{proof}

\section{Open questions}

\begin{question}
\label{question:upperBound}
Do we have, for any graph $G$, $\kiu(G) \leq \logV{V(G)} +1$ ?
\end{question}

In order to prove such a result, Lemma~\ref{lem:subgraphImpliesGraph} might be really helpful. Indeed, if one can find an edge-subgraph of $G$ for which there exists a coloring using the optimal number of colors, then one can prove that the graph itself can be colored using at most the optimal number of colors plus one. In particular, showing that one of the two following properties holds would ensure that any graph can be colored using the optimal number of colors plus one:

\begin{itemize}
\item Any tree can be optimally colored.
\item Any forest of stars subdivided at most one time can be optimally colored.
\end{itemize}
More generally, it would also be interesting to understand which properties on the graph ensure that it can be optimally colored. The sparsity of the graph may be an interesting parameter to look at.

In addition, it may be interesting to consider several usual variations of the distinguishing problem such as the variant where the coloring is proper or where only adjacent vertices must be distinguished.

\begin{center}
	\section*{References}
\end{center}

\end{document}